\documentclass[11pt,a4paper]{article}
\usepackage[utf8]{inputenc}
\usepackage[english]{babel}
\usepackage{amsmath}
\usepackage{amsfonts}
\usepackage{amssymb}
\usepackage{amsthm}
\usepackage{makeidx}
\usepackage{graphicx}

\usepackage[font=small, labelfont=bf]{caption}

\usepackage[left=3cm,right=3cm,top=3cm,bottom=3cm]{geometry} 

\usepackage{todonotes}
\usepackage{tikz}
\usetikzlibrary{decorations.pathreplacing}
\usepackage[ampersand]{easylist}

\usepackage{xcolor}
\definecolor{lightgray}{gray}{0.95}

\usepackage[numbers,sort]{natbib}
\makeatletter
\def\NAT@spacechar{~}
\makeatother

\usepackage{hyperref}
\usepackage{cleveref}
\usepackage{xspace}

\usepackage{authblk}

\theoremstyle{theorem}

\newtheorem{mytheorem}{Theorem}
\newtheorem{mycor}{Corollary}
\theoremstyle{definition}

\theoremstyle{remark}
\newtheorem{myinv}{Observation}
\newtheorem{myclaim}{Claim}

\newcommand{\prob}[5]{%
  \begingroup
  \par\medskip
  \noindent \colorbox{lightgray}{\textsc{#1}}\nopagebreak[4]
  \par\noindent\hangindent=\parindent\textit{#2}  #3
  \par\noindent\hangindent=\parindent\textit{#4}  #5
  \par  \medskip
  \endgroup
}

\newcommand{\decprob}[3]{\prob{#1}{Input:}{#2}{Question:}{#3}}

\DeclareMathOperator{\tw}{tw}

\newcommand{\N}{\mathbb{N}}

\newcommand{\msetsc}{\textsc{Minimum Shared Edges}\xspace}
\newcommand{\MSE}{MSE\xspace}
\newcommand{\pmsetsc}{\textsc{Planar Minimum Shared Edges}\xspace}
\newcommand{\PMSE}{\textsc{Planar MSE}\xspace}
\newcommand{\VC}{\textsc{VC}\xspace}

\title{The Minimum Shared Edges Problem on Planar Graphs%
}

\author[1]{Till~Fluschnik\thanks{Till~Fluschnik acknowledges support by the DFG, project DAMM (NI~369/13-2).}}
\author[1]{Manuel~Sorge\thanks{Manuel~Sorge acknowledges support by the DFG, project DAPA (NI-369/12-2).}}

\affil[1]{\small{Institut f\"ur Softwaretechnik und Theoretische Informatik, TU~Berlin, Germany, \texttt{\{till.fluschnik, manuel.sorge\}@tu-berlin.de}}}

\makeindex

\begin{document}

\maketitle
\begin{abstract}
\small{
We study the \textsc{Minimum Shared Edges} problem introduced by Omran et al. [Journal of Combinatorial Optimization, 2015] on planar graphs: \PMSE\ asks, given a planar graph $G=(V,E)$, two distinct vertices $s,t\in V$, and two integers $p,k\in \N$, whether there are $p$ $s$-$t$~paths in $G$ that share at most $k$ edges, where an edges is called shared if it appears in at least two of the $p$ $s$-$t$~paths. We show that \PMSE\ is NP-hard by reduction from \textsc{Vertex Cover}. We make use of a grid-like structure, where the alignment (horizontal/vertical) of the edges in the grid correspond to selection and validation gadgets respectively.
}
\end{abstract}

\vspace{3pt}\noindent\small{\textbf{Keywords:} 
Grids, 
Reductions, 
NP-completeness.
}

\section{Introduction}

We study the following problem.

\decprob{Planar Minimum Shared Edges (Planar MSE)}{An undirected planar graph $G=(V,E)$ with distinct vertices $s,t\in V$, and two integers $p\in \mathbb{N}$ and $k\in \mathbb{N}_0$.}{Are there $p$ $s$-$t$~paths in $G$ that share at most $k$ edges?}
\noindent Herein, an edge is called \emph{shared} if it appears in at least two $s$-$t$ paths of the solution. To clearly distinguish between paths in a solution and ordinary paths in a graph, we also call the paths in a solution \emph{$s$-$t$~routes}. 

\PMSE is the special case of \msetsc (\MSE) in which the input graph is planar. 
\MSE\ was introduced on directed graphs by~\citet{OmranSZ13}. \citet{YeLLZ13} proved that the problem is solvable in polynomial-time on graphs of bounded treewidth. 
\citet{FluschnikKNS15} proved that \MSE is fixed-parameter tractable (FPT) with respect to the number~$p$ of desired $s$-$t$~routes, that is, there is an algorithm solving \MSE\ in $f(p)\cdot n^{O(1)}$~time, where $f$ is a computable function and $n$~denotes the number of vertices in the graph of the input instance of~\MSE. 
Moreover, \citet{FluschnikKNS15} showed that \MSE\ parameterized by the treewidth~$\tw$ and the number~$k$ of shared edges combined is W[1]-hard, that is, it is unlikely that \MSE parameterized by $\tw$ and $k$ combined admits an FPT algorithm.
For more results on \MSE and related work, we refer the reader to~\citet{Flu15}. 

In this paper, we prove the following result.

\begin{mytheorem}\label{theo!1}
\pmsetsc is NP-hard, even on planar graphs of degree at most four.
\end{mytheorem}

In the proof of~\Cref{theo!1}, we reduce from the following problem.

\decprob{Vertex Cover (VC)}{An undirected graph $G=(V,E)$ and an integer $k\in \mathbb{N}$.}{Is there a subset $W\subseteq V$ of vertices in $G$ with $|W|\leq k$ such that each edge in $G$ is incident with at least one vertex in $W$?}
\VC is one of \citet{Karp72}'s 21 original NP-complete problems.
It is worth mentioning that NP-completeness proof of \MSE due to~\citet{Flu15} implies W[2]-hardness with respect to the number~$k$ of shared edges.
In contrast, our proof of~\Cref{theo!1} does not provide insight into the complexity of \PMSE parameterized by~$k$.

\section{Preliminaries}
By $\mathbb{N}$ we denote the set of positive integers excluding~$0$ and in $\mathbb{N}_0$ we include~$0$. For $n \in \mathbb{N}$ we use~$[n]$ to denote the set~$\{1, \ldots, n\} \subseteq \mathbb{N}$. We use standard graph notation, see \citet{Diestel10}, for example. For the relevant notions of parameterized complexity, we refer the reader to the literature~\cite{Nie06,FG06,DowneyF13,CyganFKLMPPS15}.

The following graphs are used in \Cref{sec:constr}, they are illustrated in \Cref{fig:subgraphs}. An \emph{$m$-chain} is a $P_{m+1}$, that is, a path with $m$~edges.  An \emph{$(\ell,m)$-bundle} is a set of $\ell$ $m$-chains with common endpoints. A \emph{$(q,\ell,m)$-feather} is an $(\ell,m)$-bundle with a $q$-chain attached to exactly one of its endpoints. We also call the $q$-chain of the feather the \emph{$q$-shaft}. An $(\ell,m)$-rainbow is a graph constructed as follows: take two paths $P_{\ell+1}^1$, $P_{\ell+1}^2$, represented as tuple of vertices $(p_1^1,\ldots,p^1_{\ell+1})$ and $(p_1^2,\ldots,p^2_{\ell+1})$, respectively, and, for all $x\in [\ell]$, connect the pair $p^1_x$, $p^2_x$ by an $m$-chain.

\begin{figure}[t]
\centering
\begin{tikzpicture}[x=0.8cm, y=0.8cm]

\node (a) at (0,0)[scale=0.75,circle,draw]{};
\node (b) at (2,0)[scale=0.75,circle,draw]{};
\draw[very thin] (a) to (b);
\draw[dotted, very thick] (a) to (b);

\node (tx1) at (1,-1.75)[]{$m$-chain};

\node (a) at (4-0.5,0)[scale=0.75,circle,draw]{};
\node (b) at (6-0.5,0)[scale=0.75,circle,draw]{};

\foreach \x in {-40,-20,...,40}{
\draw[very thin] (a) to [out=\x,in=180-\x](b);
\draw[dotted, very thick] (a) to [out=\x,in=180-\x](b);
}

\node (tx1) at (5-0.5,-1.75)[]{$(\ell,m)$-bundle};

\draw[decorate, decoration={brace, amplitude=4pt}, thin] (4-0.5,0.75)--(6-0.5,0.75) node[midway,above,label=90:{$m$}]{};
\draw[decorate, decoration={brace, amplitude=4pt}, thin] (6.25-0.5,0.5)--(6.25-0.5,-0.5) node[midway,right,label=0:{$\ell$}]{};

\node (s) at(8-0.25,0)[scale=0.75,circle, draw]{};
\node (a) at (9-0.25,0)[scale=0.33,circle, draw]{};
\node (b) at (12-0.25,0)[scale=0.75,circle,draw]{};

\draw[very thin] (s)  to (a);
\draw[dotted, very thick] (s)  to (a);

\foreach \x in {-40,-20,...,40}{
\draw[very thin] (a) to [out=\x,in=180-0.5*\x](b);
\draw[dotted, very thick] (a) to [out=\x,in=180-0.5*\x](b);
}

\node (tx1) at (10-0.25,-1.75)[]{$(q,\ell,m)$-feather};

\draw[decorate, decoration={brace, amplitude=3pt}, thin] (8-0.25,0.75)--(9-0.1-0.25,0.75) node[midway,above,label=90:{$q$}]{};
\draw[decorate, decoration={brace, amplitude=4pt}, thin] (9+0.1-0.25,0.75)--(12-0.25,0.75) node[midway,above,label=90:{$m$}]{};
\draw[decorate, decoration={brace, amplitude=4pt}, thin] (12.25-0.25,0.5)--(12.25-0.25,-0.5) node[midway,right,label=0:{$\ell$}]{};

\node (s) at(14,0)[scale=0.75,circle, draw]{};
\node (t) at (18,0)[scale=0.75,circle,draw]{};

\draw (s) to (15.75,0);
\draw (t) to (16.25,0);

\foreach \x in {0.5,1,...,3}{
\node (a) at (16-0.5*\x,0)[scale=0.33,circle,draw]{};
\node (b) at (16+0.5*\x,0)[scale=0.33,circle,draw]{};
\draw[very thin] (a) to [out=90,in=90](b);
\draw[dotted, very thick] (a) to [out=90,in=90](b);
}

\node (tx1) at (16,-1.75)[]{$(\ell,m)$-rainbow};
\draw[decorate, decoration={brace, amplitude=4pt}, thin] (15.75,-0.25)--(14,-0.25) node[midway,below,label=-90:{$\ell$}]{};

\end{tikzpicture}
\caption{Illustration of a chain, bundle, feather and rainbow.}\label{fig:subgraphs}
\end{figure}
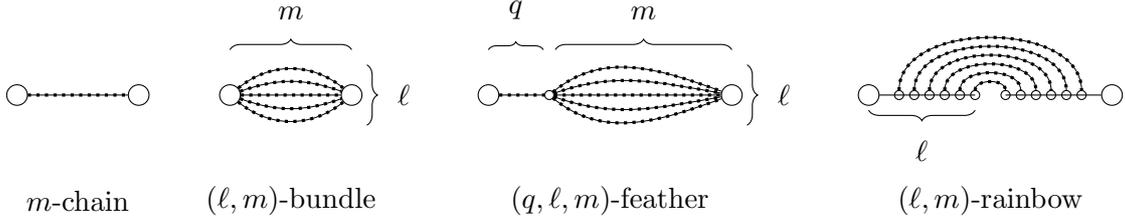

\paragraph{Grids}
For $a,b\in \mathbb{N}$, we denote by $\boxplus_{a,b}$ the graph with vertex set $\{(i,j)\mid i\in[a],j\in[b]\}$ and edge set $\{\{(i,j),(k,\ell)\}\mid |i-k|+|j-\ell|=1\}$. We call $\boxplus_{a,b}$ the $a\times b$-grid. For $i\in[a]$, we call the subgraph $R_i$ in graph $\boxplus_{a,b}$ induced by the vertex set $\{(i,j)\mid j\in[b]\}$ the \emph{row $i$}. For $j\in [b]$, we call the subgraph $C_i$ in graph $\boxplus_{a,b}$ induced by the vertex set $\{(i,j)\mid i\in[a]\}$ the \emph{column $j$}. We call the family $\{R_1,\ldots,R_a\}$ the \emph{rows} of $\boxplus_{a,b}$, and we call the family $\{C_1,\ldots,C_b\}$ the \emph{columns} of $\boxplus_{a,b}$. We denote the edges in the rows by \emph{horizontal} edges and the edges in the columns by \emph{vertical} edges.

\section{Proof of \Cref{theo!1}}\label{sec:theo!1}

We now prove \Cref{theo!1}. Given an instance of \VC{} we first describe how to obtain an equivalent instance of \PMSE{} containing a planar graph of arbitrary maximum degree. We then show how to reduce the maximum degree to at most~$4$.

\paragraph{Construction}\label{sec:constr}

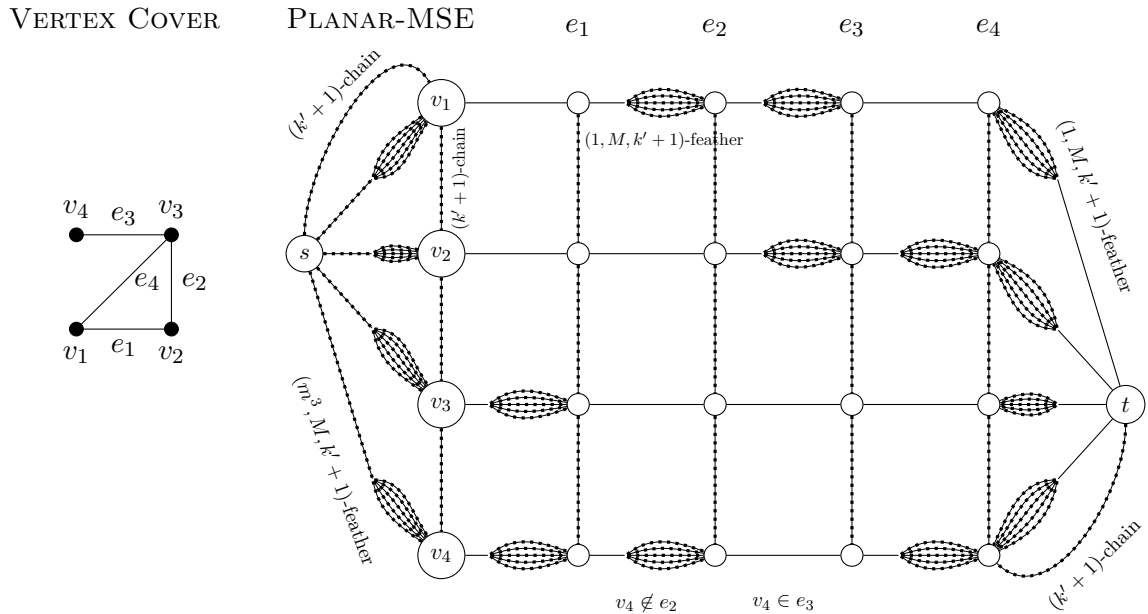
\begin{figure}[t]
\centering
\begin{tikzpicture}

\def\xl{-1};
\def\yl{-1};
\def\xr{2};

\def\vcx{1.25}
\tikzstyle{nodevc}=[circle,fill, scale=1/2, draw];

\node (v1) at (\xl+0,0+\yl)[nodevc,label=-90:{$v_1$}]{};
\node (v2) at (\xl+\vcx,0+\yl)[nodevc,label=-90:{$v_2$}]{};
\node (v3) at (\xl+\vcx,\vcx+\yl)[nodevc,label=90:{$v_3$}]{};
\node (v4) at (\xl+0,\vcx+\yl)[nodevc,label=90:{$v_4$}]{};

\draw (v1) -- node[midway, below]{$e_1$}(v2) -- node[midway, right]{$e_2$}(v3) -- node[midway, above]{$e_3$}(v4);
\draw (v1) -- node[midway, right]{$e_4$}(v3);

\node (lbl) at (\xl+.4*\vcx,2.5*\vcx)[]{\textsc{Vertex Cover}};


\tikzstyle{nodemse}=[circle, scale=1/1.25, draw];
\tikzstyle{lblnode}=[scale=0.7]
\tikzstyle{lblnodein}=[scale=0.6]

\def\msey{2}
\def\msex{1.8}
\def\txt{0.5}
\node (lbl) at (\xr+1,2.5*\vcx)[]{\textsc{Planar-MSE}};

\node (s) at (\xr+0,0)[nodemse]{$s$};
\node (t) at (\xr+6*\msex,-\msey)[nodemse]{$t$};

\node (v1) at (\xr+\msex,0+1*\msey)[nodemse]{$v_1$};
\node (v2) at (\xr+\msex,0-0*\msey)[nodemse]{$v_2$};
\node (v3) at (\xr+\msex,0-1*\msey)[nodemse]{$v_3$};
\node (v4) at (\xr+\msex,0-2*\msey)[nodemse]{$v_4$};

\node (v1e1) at (\xr+2*\msex,0+1*\msey)[nodemse]{};
\node (v2e1) at (\xr+2*\msex,0-0*\msey)[nodemse]{};
\node (v3e1) at (\xr+2*\msex,0-1*\msey)[nodemse]{};
\node (v4e1) at (\xr+2*\msex,0-2*\msey)[nodemse]{};

\node (v1e2) at (\xr+3*\msex,0+1*\msey)[nodemse]{};
\node (v2e2) at (\xr+3*\msex,0-0*\msey)[nodemse]{};
\node (v3e2) at (\xr+3*\msex,0-1*\msey)[nodemse]{};
\node (v4e2) at (\xr+3*\msex,0-2*\msey)[nodemse]{};

\node (v1e3) at (\xr+4*\msex,0+1*\msey)[nodemse]{};
\node (v2e3) at (\xr+4*\msex,0-0*\msey)[nodemse]{};
\node (v3e3) at (\xr+4*\msex,0-1*\msey)[nodemse]{};
\node (v4e3) at (\xr+4*\msex,0-2*\msey)[nodemse]{};

\node (v1e4) at (\xr+5*\msex,0+1*\msey)[nodemse]{};
\node (v2e4) at (\xr+5*\msex,0-0*\msey)[nodemse]{};
\node (v3e4) at (\xr+5*\msex,0-1*\msey)[nodemse]{};
\node (v4e4) at (\xr+5*\msex,0-2*\msey)[nodemse]{};

\node (lbl) at (\xr+2*\msex,0+1.5*\msey)[]{$e_1$};
\node (lbl) at (\xr+3*\msex,0+1.5*\msey)[]{$e_2$};
\node (lbl) at (\xr+4*\msex,0+1.5*\msey)[]{$e_3$};
\node (lbl) at (\xr+5*\msex,0+1.5*\msey)[]{$e_4$};


\node (h) at (\xr+0.5*\msex,0.5*\msey)[circle, scale=1/10]{};
\draw[very thin] (s)  to (h);
\draw[dotted, very thick] (s)  to (h);

\foreach \x in {-40,-20,...,40}{
\draw[very thin] (h) to [out=\x+40,in=225-0.25*\x](v1);
\draw[dotted, very thick] (h)to [out=\x+40,in=225-0.25*\x](v1);
}

\node (h) at (\xr+0.5*\msex,-0.0*\msey)[circle, scale=1/10]{};
\draw[very thin] (s)  to (h);
\draw[dotted, very thick] (s)  to (h);

\foreach \x in {-40,-20,...,40}{
\draw[very thin] (h) to [out=\x,in=180-0.25*\x](v2);
\draw[dotted, very thick] (h)to [out=\x,in=180-0.25*\x](v2);
}

\node (h) at (\xr+0.5*\msex,-0.5*\msey)[circle, scale=1/10]{};
\draw[very thin] (s)  to (h);
\draw[dotted, very thick] (s)  to (h);

\foreach \x in {-40,-20,...,40}{
\draw[very thin] (h) to [out=\x-40,in=135-0.25*\x](v3);
\draw[dotted, very thick] (h)to [out=\x-40,in=135-0.25*\x](v3);
}

\node (h) at (\xr+0.5*\msex,-1.5*\msey)[circle, scale=1/10]{};
\draw[very thin] (s)  to (h);
\draw[dotted, very thick] (s)  to (h);

\foreach \x in {-40,-20,...,40}{
\draw[very thin] (h) to [out=\x-40,in=135-0.25*\x](v4);
\draw[dotted, very thick] (h)to [out=\x-40,in=135-0.25*\x](v4);
}

\node (lbl) at (\xr+0.5*\msex-\txt,-1.5*\msey+0.25*\txt)[rotate=-70, lblnode]{$(m^3,M,k'+1)$-feather};

\draw[very thin] (s)  to [out=90,in=120](v1);
\draw[dotted, very thick] (s)  to [out=90,in=120](v1);

\node (lbl) at (\xr+0.5*\msex-\txt,+1.0*\msey+0.25*\txt)[rotate=45, lblnode]{$(k'+1)$-chain};


\node (h) at (\xr+5.5*\msex,0.5*\msey)[circle, scale=1/10]{};
\draw (t)  to (h);
\foreach \x in {-40,-20,...,40}{
\draw[very thin] (h) to [out=\x+180-40,in=135+180-0.5*\x](v1e4);
\draw[dotted, very thick] (h)to [out=\x+180-40,in=135+180-0.5*\x](v1e4);
}

\node (h) at (\xr+5.5*\msex,-0.5*\msey)[circle, scale=1/10]{};
\draw (t)  to (h);
\foreach \x in {-40,-20,...,40}{
\draw[very thin] (h) to [out=\x+180-40,in=135+180-0.5*\x](v2e4);
\draw[dotted, very thick] (h)to [out=\x+180-40,in=135+180-0.5*\x](v2e4);
}

\node (h) at (\xr+5.5*\msex,-1.0*\msey)[circle, scale=1/10]{};
\draw (t)  to (h);
\foreach \x in {-40,-20,...,40}{
\draw[very thin] (h) to [out=\x+180,in=180+180-0.5*\x](v3e4);
\draw[dotted, very thick] (h)to [out=\x+180,in=180+180-0.5*\x](v3e4);
}

\node (h) at (\xr+5.5*\msex,-1.5*\msey)[circle, scale=1/10]{};
\draw (t)  to (h);
\foreach \x in {-40,-20,...,40}{
\draw[very thin] (h) to [out=\x+180+40,in=225+180-0.5*\x](v4e4);
\draw[dotted, very thick] (h)to [out=\x+180+40,in=225+180-0.5*\x](v4e4);
}
\node (lbl) at (\xr+5.5*\msex+1.0*\txt,0.5*\msey-0.75*\txt)[rotate=-70, lblnode]{$(1,M,k'+1)$-feather};

\draw[very thin] (t)  to [out=-90,in=-45](v4e4);
\draw[dotted, very thick] (t)  to [out=-90,in=-45](v4e4);

\node (lbl) at (\xr+5.5*\msex+\txt,-2.0*\msey-0.25*\txt)[rotate=45, lblnode]{$(k'+1)$-chain};


\draw (v1)  to (v1e1);

\node (h) at (\xr+2.35*\msex,1*\msey)[circle, scale=1/10]{};
\draw (v1e1)  to (h);
\foreach \x in {-40,-20,...,40}{
\draw[very thin] (h) to [out=\x,in=180-0.5*\x](v1e2);
\draw[dotted, very thick] (h) to [out=\x,in=180-0.5*\x](v1e2);
}
\node (lbl) at (\xr+2.35*\msex+1*\txt,1*\msey-\txt)[rotate=0, lblnodein]{$(1,M,k'+1)$-feather};

\node (h) at (\xr+3.35*\msex,1*\msey)[circle, scale=1/10]{};
\draw (v1e2)  to (h);
\foreach \x in {-40,-20,...,40}{
\draw[very thin] (h) to [out=\x,in=180-0.5*\x](v1e3);
\draw[dotted, very thick] (h) to [out=\x,in=180-0.5*\x](v1e3);
}

\draw (v1e3) to (v1e4);


\draw (v2)  to (v2e1);
\draw (v2e1)  to (v2e2);

\node (h) at (\xr+3.35*\msex,0*\msey)[circle, scale=1/10]{};
\draw (v2e2)  to (h);
\foreach \x in {-40,-20,...,40}{
\draw[very thin] (h) to [out=\x,in=180-0.5*\x](v2e3);
\draw[dotted, very thick] (h) to [out=\x,in=180-0.5*\x](v2e3);
}

\node (h) at (\xr+4.35*\msex,0*\msey)[circle, scale=1/10]{};
\draw (v2e3)  to (h);
\foreach \x in {-40,-20,...,40}{
\draw[very thin] (h) to [out=\x,in=180-0.5*\x](v2e4);
\draw[dotted, very thick] (h) to [out=\x,in=180-0.5*\x](v2e4);
}


\draw (v3e1)  to (v3e2);
\draw (v3e2)  to (v3e3);
\draw (v3e3)  to (v3e4);

\node (h) at (\xr+1.35*\msex,-1*\msey)[circle, scale=1/10]{};
\draw (v3)  to (h);
\foreach \x in {-40,-20,...,40}{
\draw[very thin] (h) to [out=\x,in=180-0.5*\x](v3e1);
\draw[dotted, very thick] (h) to [out=\x,in=180-0.5*\x](v3e1);
}


\draw (v4e2)  to (v4e3);

\node (h) at (\xr+1.35*\msex,-2*\msey)[circle, scale=1/10]{};
\draw (v4)  to (h);
\foreach \x in {-40,-20,...,40}{
\draw[very thin] (h) to [out=\x,in=180-0.5*\x](v4e1);
\draw[dotted, very thick] (h) to [out=\x,in=180-0.5*\x](v4e1);
}

\node (h) at (\xr+2.35*\msex,-2*\msey)[circle, scale=1/10]{};
\draw (v4e1)  to (h);
\foreach \x in {-40,-20,...,40}{
\draw[very thin] (h) to [out=\x,in=180-0.5*\x](v4e2);
\draw[dotted, very thick] (h) to [out=\x,in=180-0.5*\x](v4e2);
}
\node (lbl) at (\xr+2.5*\msex,-2.5*\msey+0.75*\txt)[rotate=0, lblnode]{$v_4\not\in e_2$};

\node (h) at (\xr+4.35*\msex,-2*\msey)[circle, scale=1/10]{};
\draw (v4e3)  to (h);
\foreach \x in {-40,-20,...,40}{
\draw[very thin] (h) to [out=\x,in=180-0.5*\x](v4e4);
\draw[dotted, very thick] (h) to [out=\x,in=180-0.5*\x](v4e4);
}
\node (lbl) at (\xr+3.5*\msex,-2.5*\msey+0.75*\txt)[rotate=0, lblnode]{$v_4\in e_3$};


\draw[very thin] (v1) -- (v2) -- (v3) -- (v4);
\draw[dotted, very thick] (v1) -- (v2) -- (v3) -- (v4);
\node (lbl) at (\xr+\msex+0.5*\txt,0.5*\msey)[rotate=90,lblnodein]{$(k'+1)$-chain};

\draw[very thin] (v1e1) -- (v2e1) -- (v3e1) -- (v4e1);
\draw[dotted, very thick] (v1e1) -- (v2e1) -- (v3e1) -- (v4e1);

\draw[very thin] (v1e2) -- (v2e2) -- (v3e2) -- (v4e2);
\draw[dotted, very thick] (v1e2) -- (v2e2) -- (v3e2) -- (v4e2);

\draw[very thin] (v1e3) -- (v2e3) -- (v3e3) -- (v4e3);
\draw[dotted, very thick] (v1e3) -- (v2e3) -- (v3e3) -- (v4e3);

\draw[very thin] (v1e4) -- (v2e4) -- (v3e4) -- (v4e4);
\draw[dotted, very thick] (v1e4) -- (v2e4) -- (v3e4) -- (v4e4);

\end{tikzpicture}
\caption{Example for the reduction from \VC\ to \PMSE\ described in \Cref{sec:constr}. Left: The graph in an instance of VC containing vertices $v_1,\ldots,v_4$ and edges $e_1,\ldots,e_4$. Right: The constructed graph in the instance of \PMSE. Herein, $m=4$, $M=12$, and $k'=69\cdot k$.}\label{fig:sketch}
\end{figure}
Let $(G=(V,E),k)$ be an instance of \VC{} and let $n:=|V|$ and $m:=|E|$. We denote $V=\{v_1,\ldots,v_n\}$ and $E=\{e_1,\ldots,e_m\}$. We construct an instance~$(G',s,t,p,k')$ of \PMSE{} as follows. First, define $M:=2\cdot(m+1) + 2$, and $k':=k\cdot(m^3+m+1)$. Below we also call $k'$ the \emph{budget}. The graph~$G'$ is constructed as follows. Refer to \Cref{fig:sketch} for an illustration.

Initially, we set $G'$ to $\boxplus_{n,m+1}$. Each row in $\boxplus_{n, m+1}$ corresponds to a vertex of $G$ and each column beside the first one corresponds to an edge of~$G$. 
More precisely, row $i$ corresponds to vertex~$v_i$ and column~$j+1$ corresponds to edge~$e_j$. 
For each $i\in[n]$, if edge $e_{j}$ is not incident with vertex $v_i$, then we replace the horizontal edge in row~$i$ connecting column~$j$ and~$j+1$ by a $(1,M,k'+1)$-feather. 
We replace all vertical edges by $k'+1$-chains. We call the resulting graph~$\boxplus'$. 
By \emph{row~$i$ of $\boxplus'$} we refer to the subgraph obtained from row $i$ in $\boxplus_{n,m+1}$ by the previously described modifications. 
Similarly, by \emph{the column $j$ of $\boxplus'$} we refer to the subgraph obtained from column~$j$ in $\boxplus_{n,m+1}$ by the previously described modifications. 
Instead of horizontal and vertical edges, we talk about horizontal and vertical connections in~$\boxplus'$, meaning the graphs that replaced the corresponding edges.

We now add the new vertices~$s$ and~$t$ to~$G'$. 
We connect all vertices in column $1$ with $s$ via $(m^3,M,k'+1)$-feathers and all vertices in column $m+1$ with $t$ via $(1,M,k'+1)$-feathers. 
Herein, we merge~$s$ with the end points of the shafts of the feathers and analogously for~$t$.
Last, we add two $(k'+1)$-chains, one connecting $s$ with vertex~$(1,1)$, and the other connecting vertex~$(n,m+1)$ with $t$. Below, we call these paths \emph{validation paths}. 
We denote the finally obtained graph by~$G'$. 

Further, we set the number~$p$ of desired $s$-$t$ routes to~$k\cdot M + (n-k)+1$. This concludes the construction of the instance~$(G', s, t, p, k')$.

\paragraph{Planarity}

The $\boxplus_{n,m+1}$ is planar, feathers and chains are planar as well. Replacing (vertical/horizontal) edges in~$\boxplus_{n,m+1}$ preserves planarity. Connecting~$s$ with all vertices in column~$1$ can be done preserving the planarity, and by symmetry, the same holds for connecting~$t$ with all vertices in the column~$m+1$.

\paragraph{Correctness}\label{sec:correc}

We claim that $(G',s,t,p,k')$ is a yes-instance of \PMSE if and only if $(G,k)$ is a yes-instance of VC.

$(\Rightarrow)$: Suppose that $(G',s,t,p,k')$ is a yes-instance of \PMSE and consider a solution to~$(G',s,t,p,k')$. 
We show that we can construct a vertex cover of size~$k$ in~$G$. 

First we state observations about a solution to~$(G',s,t,p,k')$. The first observation is about~$s$, its incident chain and feathers, and how $k'$ $s$-$t$ routes determine a $k$-vertex subset of~$G$. Note that the degree of $s$ is exactly $(n+1)$. 
At most one $s$-$t$~route contains the validation paths, otherwise there are $k'+1$ shared edges, contradicting the fact that~$(G',s,t,p,k')$ is a yes-instance.
For the same reason, each $(m^3,M,k'+1)$-feather appears in at most $M$~routes. 
If a $(m^3,M,k'+1)$-feather appears in at least two routes, then $m^3$~edges are shared. 
Since the budget allows for $k\cdot m^3$ edges, at most $k$ $(m^3,M,k'+1)$-feathers appear in at least two routes each. 
Since there are $p:=k\cdot M + (n-k)+1$ $s$-$t$~routes, we obtain the following.

\begin{myinv}\label{inv:1}
In any solution to the instance $(G',s,t,p,k')$, there are exactly $k$~feathers connecting~$s$ with the vertices in the first column that contain $M$ $s$-$t$~routes each. 
All the other $n-k$~feathers contain exactly one route each. 
Moreover, the $(k'+1)$-chain incident with~$s$ appears in exactly one $s$-$t$~route.
\end{myinv}

\noindent We say that the row $i$ is \emph{selected} if the feather connecting $s$ with vertex $(i,1)$ is contained in $M$~routes. 

The second observation is about the number of shared edges in a selected row. Note that each vertical connection in~$\boxplus'$ is a $(k'+1)$-chain and, thus, none of them appears in at least two routes.
Recall that if row~$i$ is selected, then vertex~$(i,1)$ appears in at least $M=2\cdot (m+1)+2$~routes.
Since each vertical connection appears in at most one route, there are at most $2\cdot(m+1)$ routes that can ``leave'' a row via vertical connections. 
This observation together with~\Cref{inv:1} yield the following.

\begin{myinv}\label{inv:2}
In any solution to the instance~$(G',s,t,p,k')$, each selected row~$i$ induces $(m+1)$~shared edges. These shared edges appear only in the horizontal connections in the selected row~$i$ and in the feather connecting vertex~$(i,m+1)$ with~$t$.
\end{myinv}

By \Cref{inv:1} and \Cref{inv:2}, we know that in any solution to~$(G',s,t,p,k')$, there are exactly~$k$~selected rows in~$G'$ and all shared edges appear in the selected rows, in the feathers connecting $s$ with the selected rows, and in the feathers connecting $t$ with the selected rows. 
Let rows~$i_1,\ldots,i_k$ be the selected rows and let $w_1,\ldots,w_k$ be the vertices in $G$ corresponding to the selected rows. Recall that by~\Cref{inv:1} and \Cref{inv:2}, no budget is left. We claim that $W:=\{w_1,\ldots,w_k\}$ is a vertex cover in $G$.

Suppose that $W$ is not a vertex cover in $G$, that is, there is an edge $e_j$ such that $v\cap e_j=\emptyset$ for all~$v\in W$. 
We show that this induces at least one additional shared edge, contradicting the fact that $(G',s,t,p,k')$ is a yes-instance of \PMSE.
If $W$ is not a vertex cover in~$G$, then $(i_\ell,j)$ and $(i_\ell,j+1)$ are connected by a $(1,M,k'+1)$-feather for each $\ell\in[k]$. 
Observe that there are at most $M\cdot k$ routes crossing column $j$ to column $j+1$ over the feathers connecting $(i_\ell,j)$ and $(i_\ell,j+1)$ with $\ell\in[k]$. 
There are $n-k$ remaining horizontal connections to cross column~$j$ to column~$j+1$. Furthermore, all~$p$ $s$-$t$~routes appear in each column. Hence there are at least $n-k+1$ routes that cross column $j$ to column $j+1$ over the $n-k$ remaining horizontal connections (recall that $p = k \cdot M + (n - k) + 1$). 
By the pigeon-hole principle, at least one of these horizontal connections appears in at least two routes.
Since each horizontal connection is either a single edge or a $(1,M,k'+1)$-feather, the two routes induce at least one further shared edge. 
Thus, there at least $k'+1$ edges shared by the $p$~$s$-$t$~routes, which contradicts the fact that $(G',s,t,p,k')$ is a yes-instance of \PMSE. 
It follows that $W$ is a vertex cover in $G$.

$(\Leftarrow)$: Suppose that $(G,k)$ is a yes-instance of VC, and let $W\subseteq V$ be a vertex cover in $G$ with $|W|=k$. 
We show that we can construct $p$ $s$-$t$~routes in $G'$ that share $k'$ edges.

We lead $M$~routes from $s$ to each vertex $(w,1)$ in $G'$ with $w\in W$ and one route to each vertex~$(x,1)$ with $x\in V\backslash W$. 
Note that these are exactly $k\cdot M+(n-k)=p-1$ routes.
Moreover, by the construction of the routes so far, $k\cdot m^3$ edges are shared.
These shared edges appear in the $m^3$-shafts of the $k$~feathers connecting $s$ with the vertices $(w,1)$ in $G'$ with $w\in W$. 

For each row~$i\in[n]$, we lead all the routes containing vertex~$(i,1)$ from vertex~$(i,1)$ to vertex $(i,m+1)$ using only the connections in row~$i$.
Note that this construction of the routes induce $k\cdot (m+1)$ further shared edges: In feathers, only the shafts need to be shared, since we define~$M$ routes and the bundle in each feather contain~$M$ edge-disjoint paths.
Finally, for each row~$i\in[n]$, we lead all routes containing $(i,m+1)$ via the feather incident with vertex $(i,m+1)$ to $t$. 
This construction yields $k$ further shared edges, namely those in the $1$-shafts of the $k$ feathers that connect column $m+1$ with~$t$, each appearing in~$M$ $s$-$t$~routes. 
Observe that, so far, $k'$~edges are shared and, thus, no budget for sharing any further edge is left. 

So far, $p-1$ routes are constructed connecting $s$ with $t$.  
Thus, one $s$-$t$~route remains, that we call the \emph{validation route} and which we construct as follows.
First, we lead the validation route to~$(1,1)$ over the $(k'+1)$-chain connecting $s$ with vertex~$(1,1)$.
Next, we route the validation route through $\boxplus'$ as follows. 
Since $W$~is a vertex cover in~$G$, for each edge $e_j$ there exists a vertex~$v_i$ in~$W$ such that~$v_i\in e_j$. 
Thus, by construction of $G'$ it holds that $(i,j)$ and $(i,j+1)$ are connected by an edge. 
By the construction of the $p-1$ $s$-$t$~routes before, it holds that the edge connecting $(i,j)$ with $(i,j+1)$ is shared by exactly $M$~paths. 
Thus, we can find in every column~$j\in[m]$ exactly one index~$i_j$ corresponding to row~$i_j$ such that $(i_j,j)$ and $(i_j,j+1)$ are connected by an edge that is shared by $M$ $s$-$t$~routes. 
We lead the validation route in each column~$j\in[m]$ to the row~$i_j$ using the vertical connections, and then over the shared edge $\{(i_j,j),(i_j,j+1)\}$ to column~$j+1$. 
In column~$m+1$, we lead the validation route via the vertical connections to $(n,m+1)$. 
Finally, we lead the validation route over the $(k'+1)$-chain connecting $(n,m+1)$ with $t$ to $t$.
Note that in the construction of the validation route, we do not share any additional edge.

We constructed $p$ $s$-$t$~routes sharing $k'$ edges in $G'$ and, thus, $(G',s,t,p,k')$ is a yes-instance of \PMSE.
\paragraph{Maximum degree at most four}
We now make the following modifications to the graph~$G'$ and the budget~$k'$ in the instance of \PMSE\ constructed above. We first replace each bundle in the graph~$G'$ by a rainbow. (For this to yield an equivalent instance, we need to subdivide edges in $G'$ before.) Then, we replace the high-degree vertices~$s$ and~$t$ by binary trees, in a similar fashion as done by \citet[Theorem 5.2]{Flu15}.

We aim to replace each bundle in~$G'$ by a rainbow. Since replacing a bundle by a rainbow may introduce additional shared edges in a solution, we have to increase the budget~$k'$ as well. However, increasing the budget may allow to share new edges outside of rainbows which we did not intend to be shareable. To circumvent this issue, we first subdivide each edge in~$G'$ several times and, only after the subdivision, replace bundles by rainbows. Subdividing edges has the effect that the number of shared edges in any solution is a multiple of the number~$b$ of subdivisions. Hence, if $b$ is larger than the increase of the budget when replacing the bundles by rainbows, no further edges other than the ones in the rainbows can be shared. We now formalize this approach.

We introduce the following notation. Let $H$ be a graph. A \emph{proper chain} in $H$ is an induced path~$P$ in~$H$ such that all inner vertices of~$P$ have degree exactly two in~$H$. If the endpoints of~$P$ each have degree different from two in~$H$, then we call $P$ a \emph{maximal} proper chain. 

To replace bundles by rainbows, we use the following claim. 
\begin{myclaim}\label{claim:rainbow}
  Let $H$ be a graph such that each maximal proper chain in~$H$ has length at least~$b$. 
  Let~$c$ be the number of $(a, d)$-bundles in $H$, such that~$2ac < b$, and let~$p, k$ be integers such that $d > k$. 
  Let $H'$ be the graph obtained from~$H$ by replacing each $(a, d)$-bundle by a $(a, d + 2ac)$-rainbow. 
  Graph~$H$ admits $p$ $s$-$t$ routes with at most~$k$ shared edges if and only if $H'$ admits $p$ $s$-$t$ routes with at most~$k + 2ac$ shared edges. 
\end{myclaim}
\begin{proof}
  $(\Rightarrow)$: Assume that $H$ admits $p$ $s$-$t$ routes with at most~$k$ shared edges. 
  Note that, since $d > k$, there are no two routes that share one of the $d$-chains in any bundle. 
  Transform this set of routes into a set of routes in~$H'$ as follows. 
  Outside of any rainbow, each route in~$H$ equals one route in~$H$. 
  Inside a rainbow, lead each route over a distinct $d$-chain. 
  In this way, at most~$c \cdot 2a$ more edges are shared. 
  That is, $H'$ admits $p$ $s$-$t$ routes with at most $k + 2ac$ shared edges.
  
  $(\Leftarrow)$: Assume that $H'$ admits $p$ $s$-$t$ routes with at most $k + 2ac$ shared edges. 
  Note that, since $d + 2ac > k + 2ac$, there are no two routes that share one of the $d$-chains in any rainbow. 
  Transform a corresponding set of routes in~$H'$ into a set of routes in~$H$ as follows. 
  Outside of any bundle, each route in~$H$ equals one route in~$H'$. 
  Inside a bundle, lead each route over a distinct $d$-chain. 
  This yields a set of $s$-$t$ routes in~$H$ with at most~$k + 2ac$ shared edges. 
  Since each maximal proper chain in $H$ has length at least $b > 2ac$ and since each maximal proper chain is either shared completely or not at all, indeed, there are at most~$b \lfloor (k + 2ac) / b \rfloor  \leq k$ shared edges.
\end{proof}

Now consider $G'$ and $k'$ from the instance of \PMSE. Note that each bundle in $G'$ is an $(M, k' + 1)$-bundle. We now to replace each of these bundles by rainbows using \Cref{claim:rainbow}. To satisfy the precondition of \Cref{claim:rainbow}, we need that each maximal proper chain in~$G'$ has length at least~$b' > 2Mc'$, where $c'$ is the number of $(M, k' + 1)$-bundles in~$G'$. For this, we perform the operation ``subdivide each edge in~$G$ and multiply $k'$ by two'' sufficiently often. 
Note that this operation yields an equivalent instance, because each $P_3$ resulting from subdividing an edge has to be traversed either completely or not at all by each route. 
Furthermore, each application of this operation doubles the minimum length~$b$ of a maximal proper chain in~$G'$, meaning that, after $O(\log(Mc'))$ applications, we have~$2Mc' < b$. Note that subdividing and multiplying~$k'$ by two does not invalidate the property that each bundle is a $(M', d')$-bundle for some $d' > k'$. Thus, we can replace all bundles in~$G'$ by rainbows and $k'$ by $k' + 2Mc'$, yielding an equivalent instance. 
Clearly, none of the above operations increases the degree of any vertex. 
On the contrary, after all operations have been applied, each vertex except~$s$ and~$t$ has degree at most 4.

To decrease the degree of $s$ and~$t$, we replace $s$ and $t$ by a complete binary trees as follows. 
Recall that the number of neighbors of $s$ and $t$ is $n + 1$ each. 
Assume that the number~$n$ of vertices in the instance of \VC\ is such that $n + 1$ is a power of two. 
Otherwise, add degree-zero vertices until this is the case. 
Replace $s$ with a complete binary tree with root~$s$ and $(n + 1)/2$ leaves. 
Make incident each previous neighbor of~$s$ in~$G'$ with one of the leaves of the complete binary tree in such a way that each leaf has degree exactly three. 
Replace $t$ with a complete binary tree in the same way, and replace~$k'$ by $k' + 2\ell$, where $\ell$ is the number of edges in a complete binary tree with $(n + 1)/2$ leaves. 
To see that the resulting instance is yes if and only if the original instance is yes, note that in any solution for~$G'$, each row of $\boxplus'$ receives at least one $s$-$t$ route from~$s$ and sends at least one $s$-$t$ route to~$t$. 
Thus, in any solution all $2\ell$ edges in the complete binary trees are shared. 
Finally, it is clear that each of the above operations can be performed while maintaining planarity.

\subsection{Directed graphs}

Let $G''$ be the planar graph of maximum degree four constructed in \cref{sec:theo!1}.
We now sketch how to modify graph~$G''$ in such a way that an equivalent instance of \pmsetsc{} on directed planar graphs is obtained.
To this end, we direct all edges except those in vertical connections from ``left to right'' with respect to a drawing as in~\Cref{fig:sketch}.
Herein, we direct the edges in the binary tree containing vertex~$s$ from~$s$ to the leaves, and the edges in the binary tree containing vertex~$t$ from the leaves to~$t$.
We replace each vertical connection by a directed graph gadget as follows (refer to~\Cref{fig:dirplanargr} in the following).
\begin{figure}[t]
\centering
  \begin{tikzpicture}

  \usetikzlibrary{decorations.pathreplacing}
  \usetikzlibrary{decorations.pathmorphing}

  \def\x{8};
  \def\xs{0.75};
  \def\y{1.75};
  \def\sx{1.7};
  \tikzstyle{xnode}=[circle, scale=0.8, draw];
  \tikzstyle{xxnode}=[circle, scale=0.5, draw];

  \node (v) at (0,\y)[xnode]{$v$};
  \node (w) at (0,-\y)[xnode]{$w$};
  \node (ld) at (0,\y+\xs)[]{$\vdots$};
  \node (ld) at (0-\xs,\y)[]{$\cdots$};
  \node (ld) at (0+\xs,\y)[]{$\cdots$};
  \node (ld) at (0,-\y-\xs)[]{$\vdots$};
  \node (ld) at (0-\xs,-\y)[]{$\cdots$};
  \node (ld) at (0+\xs,-\y)[]{$\cdots$};

  \draw[ultra thin] (v) -- (w);
  \draw[dotted, ultra thick] (v) -- (w);

  \draw[decorate, decoration={brace, amplitude=4pt}, thin, color=gray] (0-0.5*\xs,-\y+0.3)--(0-0.5*\xs,\y-0.3) node[midway,label=180:{$(k'+1)$-chain}]{};

  \draw[decorate, decoration={snake, segment length=3mm, amplitude=0.5mm},->,>=stealth] (0.4*\x-0.75,0) -- (0.4*\x+0.75,0);

  \node (v) at (0+\x,\y)[xnode]{$v$};
  \node (w) at (0+\x,-\y)[xnode]{$w$};
  \node (ld) at (0+\x,\y+\xs)[]{$\vdots$};
  \node (ld) at (0-\xs+\x,\y)[]{$\cdots$};
  \node (ld) at (0+\xs+\x,\y)[]{$\cdots$};
  \node (ld) at (0+\x,-\y-\xs)[]{$\vdots$};
  \node (ld) at (0-\xs+\x,-\y)[]{$\cdots$};
  \node (ld) at (0+\xs+\x,-\y)[]{$\cdots$};

  \node (s1) at (0+\x-\sx,0)[xxnode,label=180:{$a_{vw}$}]{};
  \node (s2) at (0+\x+\sx,0)[xxnode,label=0:{$b_{vw}$}]{};

  \draw[->,>=stealth] (v) to (s1);
  \draw[->,>=stealth] (w) to (s1);
  \draw[->,>=stealth] (s2) to (v);
  \draw[->,>=stealth] (s2) to (w);
  \draw[->,>=stealth] (s1) to (s2);

  \foreach \z in {0.3,0.45,...,3.35}{
  \draw[->,>=stealth] (\x-\sx+\z-0.15,0) to (\x-\sx+\z,0);
  }

  \draw[decorate, decoration={brace, amplitude=4pt}, thin, color=gray] (\x-\sx+0.2,0+0.3)--(\x+\sx-0.2,0+0.3) node[midway,label=90:{directed $(k'+1)$-chain}]{};

  \end{tikzpicture}
\caption{Replacement of a undirected vertical connection by the directed graph gadget.}\label{fig:dirplanargr} 
\end{figure}
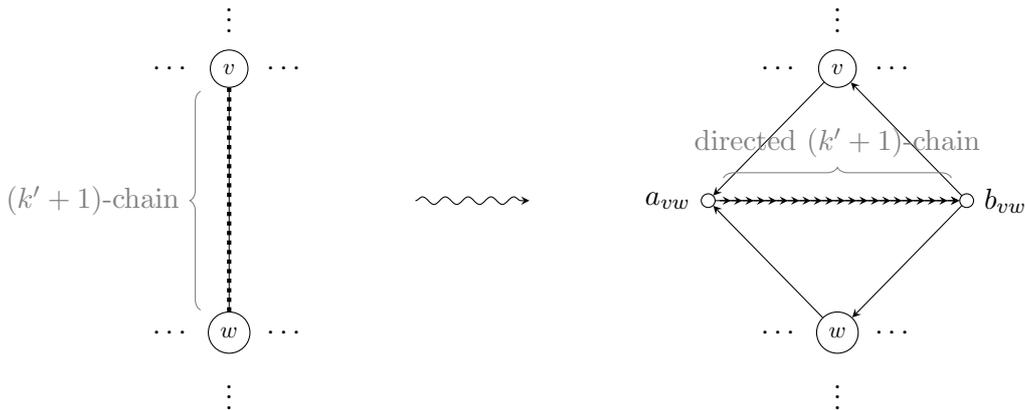

Consider a vertical connection between two vertices $v$ and $w$ (recall that $v$ and $w$ are in consecutive rows).
Remove the $(k'+1)$-chain, and add two vertices $a_{vw}$ and $b_{vw}$.
Connect $a_{vw}$ with $b_{vw}$ via a $(k'+1)$-chain, and direct all edges towards $b_{vw}$.
Finally, add the arcs $(v,a_{vw})$ and $(w,a_{vw})$, as well as the arcs $(b_{vw},v)$ and $(b_{vw},w)$. 
We apply this to each vertical connection in $G''$.

Observe that no two routes can traverse a gadget without sharing at least $k'+1$ arcs.
Moreover, any route in row~$i$, $1\leq i<n$, can traverse through each gadget in each column to row~$i+1$, and vice versa.
Thus, the introduced gadgets work like the vertical $(k'+1)$-chains in~$G''$.

Analogously to the proof of~\Cref{theo!1} we obtain the following.

\begin{mycor}
 \pmsetsc{} is NP-hard on directed planar graphs of maximum out- and indegree three.
\end{mycor}

\section{Conclusion}
We proved that \pmsetsc is NP-hard, even on planar graphs of maximum degree four, leading to the natural question, whether the problem remains NP-hard on planar graphs of maximum degree three. Furthermore, we showed that \pmsetsc is NP-hard on directed planar graphs of maximum out- and indegree three.

From a parameterized complexity perspective, since our reduction is not a parameterized reduction with respect to the number~$k$ of shared edges, the question about the parameterized complexity of \PMSE\ parameterized by~$k$ remains open.
Moreover, it would be interesting to see whether the running times of known FPT-algorithms~\cite{YeLLZ13,FluschnikKNS15} for \msetsc{} can be improved for \pmsetsc.


%
\bibliographystyle{plainnat}
\bibliography{MSE_planar_bib}

\begin{thebibliography}{10}
\providecommand{\natexlab}[1]{#1}
\providecommand{\url}[1]{\texttt{#1}}
\expandafter\ifx\csname urlstyle\endcsname\relax
  \providecommand{\doi}[1]{doi: #1}\else
  \providecommand{\doi}{doi: \begingroup \urlstyle{rm}\Url}\fi

\bibitem[Cygan et~al.(2015)Cygan, Fomin, Kowalik, Lokshtanov, Marx, Pilipczuk,
  Pilipczuk, and Saurabh]{CyganFKLMPPS15}
Marek Cygan, Fedor~V. Fomin, Lukasz Kowalik, Daniel Lokshtanov, D{\'{a}}niel
  Marx, Marcin Pilipczuk, Michal Pilipczuk, and Saket Saurabh.
\newblock \emph{Parameterized Algorithms}.
\newblock Springer, 2015.

\bibitem[Diestel(2010)]{Diestel10}
Reinhard Diestel.
\newblock \emph{Graph Theory}, volume 173 of \emph{Graduate Texts in
  Mathematics}.
\newblock Springer, 4th edition, 2010.

\bibitem[Downey and Fellows(2013)]{DowneyF13}
Rodney~G. Downey and Michael~R. Fellows.
\newblock \emph{Fundamentals of Parameterized Complexity}.
\newblock Texts in Computer Science. Springer, 2013.

\bibitem[Flum and Grohe(2006)]{FG06}
J{\"o}rg Flum and Martin Grohe.
\newblock \emph{Parameterized Complexity Theory}.
\newblock Springer, 2006.

\bibitem[Fluschnik(2015)]{Flu15}
Till Fluschnik.
\newblock The parameterized complexity of finding paths with shared edges.
\newblock Master thesis, In\-sti\-tut f\"ur Softwaretechnik und Theoretische
  Informatik, TU~Berlin, 2015.
\newblock
  \url{http://fpt.akt.tu-berlin.de/publications/theses/MA-till-fluschnik.pdf}.

\bibitem[Fluschnik et~al.(2015)Fluschnik, Kratsch, Niedermeier, and
  Sorge]{FluschnikKNS15}
Till Fluschnik, Stefan Kratsch, Rolf Niedermeier, and Manuel Sorge.
\newblock The parameterized complexity of the minimum shared edges problem.
\newblock In \emph{Proc.\ of the 35th {IARCS} Annual Conference on Foundation
  of Software Technology and Theoretical Computer Science ({FSTTCS} 2015)},
  pages 448--462, 2015.

\bibitem[Karp(1972)]{Karp72}
Richard~M. Karp.
\newblock Reducibility among combinatorial problems.
\newblock In \emph{Proc.\ of a symposium on the Complexity of Computer
  Computations}, pages 85--103, 1972.

\bibitem[Niedermeier(2006)]{Nie06}
Rolf Niedermeier.
\newblock \emph{Invitation to Fixed-Parameter Algorithms}.
\newblock Oxford University Press, 2006.

\bibitem[Omran et~al.(2013)Omran, Sack, and Zarrabi{-}Zadeh]{OmranSZ13}
Masoud~T. Omran, J{\"{o}}rg{-}R{\"{u}}diger Sack, and Hamid Zarrabi{-}Zadeh.
\newblock Finding paths with minimum shared edges.
\newblock \emph{Journal of Combinatorial Optimization}, 26\penalty0
  (4):\penalty0 709--722, 2013.

\bibitem[Ye et~al.(2013)Ye, Li, Lu, and Zhou]{YeLLZ13}
Zhi-Qian Ye, Yi-Ming Li, Hui-Qiang Lu, and Xiao Zhou.
\newblock Finding paths with minimum shared edges in graphs with bounded
  treewidths.
\newblock In \emph{Proc. Frontiers of Computer Science (FCS '13)}, pages
  40--46, 2013.

\end{thebibliography}

\end{document}